\documentclass[%
 aip,
 amsmath,amssymb,
reprint, onecolumn 
]{revtex4-1}

\usepackage{graphicx}
\usepackage{dcolumn}
\usepackage{bm}

\usepackage[utf8]{inputenc}
\usepackage[T1]{fontenc}
\usepackage{mathptmx}
\usepackage{etoolbox}

 \usepackage{amsthm}
\newtheorem{mydef}{Definition}
\newtheorem{mythm}{Theorem}
\newtheorem{mylemm}{Lemma}
\newtheorem{mycorol}{Corollary}

\usepackage{hyperref}
\makeatletter
\def\@email#1#2{%
 \endgroup
 \patchcmd{\titleblock@produce}
  {\frontmatter@RRAPformat}
  {\frontmatter@RRAPformat{\produce@RRAP{*#1\href{mailto:#2}{#2}}}\frontmatter@RRAPformat}
  {}{}
}%
\makeatother
\begin{document}

\preprint{AIP/123-QED}

\title{A Simple and General Equation for Matrix Product Unitary Generation}
\author{Sujeet K. Shukla}
 \email{sshukla@alumni.caltech.com.}
 \affiliation{Institute for Quantum Information and Matter, California Institute of Technology, California, USA}

\date{\today}

\begin{abstract}
Matrix Product Unitaries (MPUs) have emerged as essential tools for representing locality-preserving 1D unitary operators, with direct applications to quantum cellular automata and quantum phases of matter. A key challenge in the study of MPUs is determining when a given local tensor generates an MPU, a task previously addressed through fixed-point conditions and canonical forms, which can be cumbersome to evaluate for an arbitrary tensor. In this work, we establish a simple and efficient necessary and sufficient condition for a tensor $M$ to generate an MPU of size $N$, namely, $\operatorname{Tr}(\mathbb{E}_M^N) = \operatorname{Tr}(\mathbb{E}_T^N) = 1$, where $\mathbb{E}_M$ and $\mathbb{E}_T$ denote the transfer matrices of $M$ and $T = M M^\dagger$, respectively, with $M^\dagger$ being the Hermitian adjoint of M with respect to the physical indices. This condition provides a unified framework for characterizing all uniform MPUs and significantly simplifies their evaluation. Furthermore, we show that locality preservation naturally arises when the MPU is generated for all system sizes. Our results offer new insights into the structure of MPUs, highlighting connections between unitary evolution, transfer matrices, and locality-preserving behavior, with potential extensions to higher-dimensions.
\end{abstract}

\maketitle

\section{Introduction}
The matrix product formalism~\cite{Fannes1992,Perez-Garcia2007} has played a foundational role in the study of one-dimensional quantum systems, providing a versatile framework for both analytical and numerical approaches. In particular, the matrix product representation of quantum states serves as the backbone of powerful computational techniques, such as the Density Matrix Renormalization Group (DMRG)~\cite{White1993}  and the Time-Evolving Block Decimation (TEBD) algorithms~\cite{Vidal2003}, which have been instrumental in exploring low-energy properties and dynamical behavior of 1D systems~\cite{Schuch2010,landau2015polynomial}. Beyond its computational advantages, the matrix product representation offers deep structural insights into quantum many-body states, facilitating rigorous proofs of the efficiency of variational algorithms in approximating ground states. Additionally, it has enabled a comprehensive classification of gapped phases in one dimension~\cite{Pollmann2010,Pollmann2012,Chen2011,Schuch2011}, revealing fundamental connections between entanglement properties, symmetries, and universality in quantum phases of matter.

Just as quantum states can be efficiently represented using matrix product states (MPS), operators acting on one-dimensional systems can be described within the matrix product operator (MPO) framework~\cite{Verstraete2004,Pirvu2010,Haegeman2016}. Not only has this formalism proven invaluable in the simulation of mixed states and real or imaginary time evolution in 1D quantum systems, it has also been found to be an extremely useful tool in understanding 2D topological quantum phases~\cite{Sahinoglu2021,Shukla16}. Of particular interest are MPOs that are unitary, or Matrix Product Unitaries (MPU), as they play a fundamental role in understanding and simulating dynamical processes while preserving entanglement area laws. MPUs have been found to describe the boundary physics of Floquet phases~\cite{Po2016} and play a key role in classifying them. There has been a series of work on MPUs in the last few years. Refs.~\onlinecite{Burak2018, Cirac2017} established fundamental properties of MPUs. They showed that MPUs are locality-preserving, which map local operators to local operators. And in fact all such unitaries in 1D can be represented as MPUs. Uniform MPU (constructed from a single repeated tensor) coincide with quantum cellular automata (QCA), which exhibit exact light-cone dynamics and are known to be implementable as finite-depth quantum circuits. This connection has led to a rich classification of 1D MPU based on their information transmission properties, characterized by the GNVW (for Gross, Nesme, Vogts and Werner~\cite{Gross2012}) index, and equivalently, \textit{rank-ratio} index~\cite{Burak2018}. There have been further explorations into MPUs by applying symmetries~\cite{Zhang2023,Gong2020}, and considering non-uniform MPUs with arbitrary boundary conditions more recently~\cite{Styliaris2024}. MPUs have also been used recently as disentanglers to prepare valence-bond-solid states on gate-based quantum computers~\cite{Murta2023}. MPUs have also been extended to fermionic systems~\cite{Piroli2021}, where they exhibit significant differences from their bosonic counterparts, reflecting the unique algebraic and topological properties of fermionic systems.

Previous formalism~\cite{Burak2018, Cirac2017} on MPUs used the approach of utilizing the canonical or standard form of the local tensor and deriving a set of fixed-point conditions for it. Various properties of MPUs, including locality-preservation, were then derived and rigorously proven based on these conditions. However, evaluating these conditions in practice for a given tensor can be challenging and often involves cumbersome calculations. For example, Ref.~\onlinecite{Burak2018} (Lemma 1) establishes an expected form for the tensor  $T = MM^{\dagger}$, where  $M$  is an MPU-generating tensor, but this condition is not intuitive and can be difficult to confirm for an arbitrary tensor. 

Moreover,  previous results do not provide a clear and efficient characterization of the sufficient and necessary conditions that a tensor must satisfy to generate an MPU. They also do not address how MPU size affects their generation and behavior. For example, how can we evaluate whether a given tensor will generate an MPU of a \textit{given} size? For instance, Ref.~\onlinecite{Burak2018} identified tensors that generate MPUs \textit{only} for odd system sizes which do not fit into the forms expected for locality-preserving MPUs. This raises the question: How can we establish a unified condition that applies to all such MPU generating tensors?

In this work, we propose a simple method for determining whether a given tensor generates an MPU of size  $N$. We present the rather surprising result that
\[
\operatorname{Tr}(\mathbb{E}_M^N) = \operatorname{Tr}(\mathbb{E}_T^N) = 1
\]
is both a necessary \textit{and sufficient} condition for a tensor $M$  to generate an MPU of size  $N$. Here, $T = MM^{\dagger}$ with $M^\dagger$ being the Hermitian adjoint of M with respect to the physical indices, and  $\mathbb{E}_M$ and  $\mathbb{E}_T$  are the corresponding transfer matrices. More details on this will be presented in later sections. This condition allows us to identify \textit{all} possible uniform MPUs of a given size by solving these equations. Furthermore, we demonstrate that the condition for locality preservation naturally follows when we require the tensor  $M$  to generate MPUs for \textit{all} system sizes.

The paper is organized as follows. In section~\ref{Prel}, we introduce the necessary notations and preliminary concepts to set the stage for our results and proofs. In section~\ref{sec:main theorem}, we present the main result of the paper in the form of an MPU-generation theorem and provide its proof. Section~\ref{sec:revisit LP MPU} addresses locality preservation as a special case of the main result. Finally, in section~\ref{sec:examples}, we illustrate some numerical examples, and we conclude the paper in section~\ref{sec:conclusion}.
\section{Preliminaries and notation} \label{Prel}
We first set up terminology and notation for this paper. A \textit{local-tensor} $M$ is a 4-index tensor with two virtual indices and two physical indices (an input and an output). We will leverage graphical representation of tensors and tensor networks throughout this work. A local tensor is represented as,
\begin{equation}\label{local-tensor}
   M^{ij}_{ab} = \raisebox{-.4\height}{\includegraphics[scale=0.3]{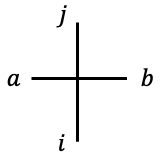}} \in \mathbb{C},
\end{equation}
where $i$ and $j$ are the input and output physical indices and $a$ and $b$ are the left and right virtual indices. Throughout this paper, we use $d$ to denote the physical dimension and $D$ to denote the virtual (bond) dimension, unless stated otherwise. Contracting $N$ such tensors gives us a Matrix Product Operator (MPO) of length N, which we will denote as $MPO^{(N)}(M)$,
\begin{eqnarray}\label{ONM}
   MPO^{(N)}(M) &=&  \sum_{ \lbrace i_n,j_n \rbrace} \operatorname{Tr}(M^{i_1j_1}M^{i_2j_2}\ldots M^{i_n,j_n})|i_1,\ldots,i_n \rangle \langle j_1,\ldots,j_n| \nonumber\\
   &=& \raisebox{-0.4\height}{\includegraphics[scale=0.3]{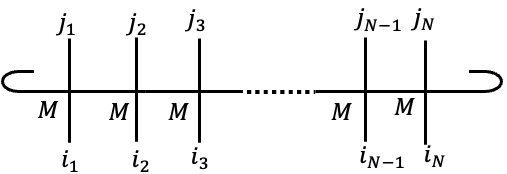}} ,
\end{eqnarray}
where the twists at both virtual ends denote periodic boundary conditions. It should be noted, and should be clear from the expression above, that throughout this paper, graphical tensor diagrams are interpreted as operators acting from bottom to top, whereas algebraic expressions are operators acting from left to right. This choice ensures consistency between the tensor network contractions and their corresponding algebraic forms. 

We define $M^{\dagger}$ to be the Hermitian adjoint of $M$ with respect to the physical dimension, 
\begin{equation}
    M^{\dagger ji} = (M^{ij})^*.
\end{equation}
It can be shown easily that if $M$ generates $MPO^{(N)}$ then $M^{\dagger}$ generates its Hermitian adjoint $MPO^{(N)\dagger}$,
\begin{equation}\label{mdaggermpo}
    MPO^{(N)}(M^\dagger) = MPO^{(N)\dagger}(M).
\end{equation}
We will see below how the same local-tensor can generate MPOs  which are unitary for some lengths $N$ and not unitary for other lengths. So we call a local-tensor $M$  a size $N$ MPU generator or, \textit{a }$MPU^{(N)}$-\textit{generator} if $MPO^{(N)}(M)$ is unitary. 

 We will see that the \textit{locality-preserving} MPUs discussed in Ref.~\onlinecite{Burak2018, Cirac11}  are generated by special local-tensors which generate unitary for \textit{all} system sizes.  We will also see examples of local tensors which generate unitaries \textit{only} for odd system sizes (see section~\ref{sec:examples}). \par 
We find it convenient to define a $T(M)$ tensor for any given local-tensor $M$,
\begin{equation}\label{Tij4}
T^{ij}(M) = \sum_{k} M^{ik}\otimes (M^{\dagger})^{kj} =\\  \raisebox{-.4\height}{\includegraphics[height=1.3cm]{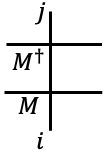}}.
\end{equation}

When the underlying $M$ is clear from the context, we would simply refer to it as the $T$-tensor. We will see that tensor $T$ plays a key role in understanding the unitary properties of $M$. \par 
For any local-tensor we can define a \textit{transfer matrix}, 
\begin{eqnarray}\label{EMdef}
\mathbb{E}_M = \frac{1}{d} \operatorname{Tr}(T(M)) = \frac{1}{d}\sum_{ij}M^{ij}M^{ij*} 
=\frac{1}{d}\raisebox{-.4\height}{\includegraphics[scale=0.3]{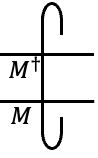}},
\end{eqnarray}
where $d$ is the dimension of the local physical Hilbert space. Applying this definition to $T$ tensor, transfer matrix of $T$ turns out to be,
\begin{equation}\label{ETdef}
\mathbb{E}_T = \frac{1}{d} \sum_{ij}T^{ij}T^{ij*} = \frac{1}{d}\raisebox{-.4\height}{\includegraphics[scale=0.3]{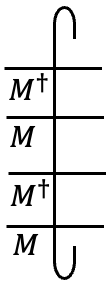}}.
\end{equation}
Notice that $T$ satisfies $T^{\dagger} = T$, but $T(M^{\dagger}) \neq T^{\dagger}(M)$. Also one can see  $\mathbb{E}_M  \equiv \mathbb{E}_{M^{\dagger}} $ and $ \mathbb{E}_{T(M)} \equiv \mathbb{E}_{T(M^{\dagger})}$ in the sense that they are the same matrices with the virtual indices permuted. 
We will now see that $\mathbb{E}_M $ and $\mathbb{E}_T $ play a key role in determining $MPU^{(N)}$-generation. \par 


\section{Necessary and sufficient condition for MPU-generation} \label{sec:main theorem}
Now we present the central result of this work as a theorem on $MPU^{(N)}$-generation. 
\begin{mythm}[$MPU^{(N)}$-generation theorem]\label{main theorem}
A local-tensor $M$ generates unitary of size $N$ if and only if 
\begin{eqnarray}\label{Nunitarycondition}
\operatorname{Tr}(\mathbb{E}_M^N) = \operatorname{Tr}(\mathbb{E}_T^N) = 1,
\end{eqnarray}
 where $\mathbb{E}_M$ and $\mathbb{E}_T$ are the transfer matrices of $M$ and $T = MM^\dagger$ (as defined in Eq.~\eqref{Tij4}).
\end{mythm}

\begin{proof} 
(i) First we prove: $MPU^{(N)}$-generation  $\Rightarrow \operatorname{Tr}(\mathbb{E}_M^N) = \operatorname{Tr}(\mathbb{E}_T^N) = 1$. \par 
This is straightforward to prove. First we note that $MPO^{(N)} MPO^{(N)\dagger }$ can be written in terms of $T$ tensors,
\begin{eqnarray}\label{eq:MPO in T}
   MPO^{(N)}MPO^{(N)\dagger }  =\raisebox{-0.4\height}{\includegraphics[scale=0.3]{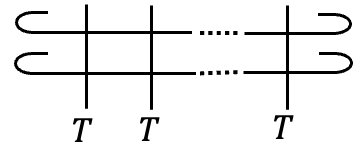}}.
\end{eqnarray}
Now, since $M$ generates MPU of length $N$, we have
\begin{equation}
\raisebox{-.6\height}{\includegraphics[scale=0.3]{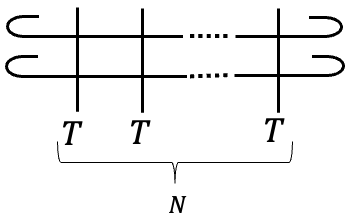}} = \raisebox{-.6\height}{\includegraphics[scale=0.3]{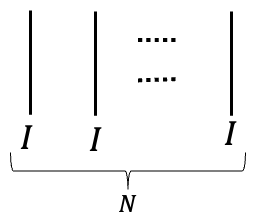}},
\end{equation}
where $I$ is the identity matrix on a given physical index. Now we take trace on the physical indices on both sides, 
\begin{eqnarray}
\raisebox{-.4\height}{\includegraphics[scale=0.3]{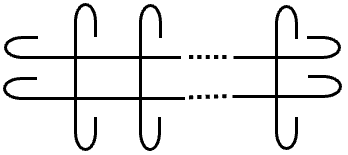}} &=& \raisebox{-.4\height}{\includegraphics[scale=0.3]{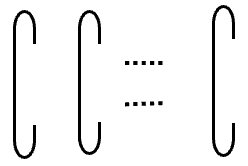}} \nonumber \\
\Rightarrow \operatorname{Tr}( (d\mathbb{E}_M)^N) &=& d^N \nonumber \\ 
\Rightarrow \operatorname{Tr}( \mathbb{E}_M^N) &=&1 .
\end{eqnarray}
Note that since $M$ generates unitary on $N$ sites, $T(M)$ also generates unitary on $N$ sites. In fact it generates identity.  Since above equation is true for any $MPU^{(N)}$-generating local tensor, it must be true for $T(M)$ as well, and hence we have $\operatorname{Tr}(\mathbb{E}_T^N )= 1$.  So we have proved $\operatorname{Tr}(\mathbb{E}_M^N ) = \operatorname{Tr}(\mathbb{E}_T^N )= 1$.

(ii) Now we prove: $\operatorname{Tr}(\mathbb{E}_M^N) = \operatorname{Tr}(\mathbb{E}_T^N) = 1 \Rightarrow MPO^{(N)\dagger}MPO^{(N)}=MPO^{(N)}MPO^{(N)\dagger}= I^{\otimes N}$. That is,  $M$ is an $MPU^{(N)}$-generating tensor. \par 
As we mentioned in Eq.~\eqref{eq:MPO in T},  $MPO^{(N)}MPO^{(N)\dagger}$ can be written in terms of $T$ tensors. So to understand unitarity of an MPO, we need to work with $T$ tensor properties.
Without loss of generality, we choose an orthonormal basis for the space of operators on a $d$-dimensional Hilbert space as $\lbrace \sigma_j , j=0,1,\ldots d^2-1\rbrace $, where $\sigma_0 = I$ and $\sigma_{j\neq 0}$ are \textit{traceless} operators with the orthonormal relationship, 
\begin{eqnarray}
\operatorname{Tr}(\sigma_j \sigma_k^{\dagger} ) = d\delta_{j,k}, \quad \forall j,k.
\end{eqnarray}
Now we expand $T(M)$ in terms of these operators on the physical space and operators on the virtual space, 
\begin{eqnarray}\label{Tsigmas}
\raisebox{-.5\height}{\includegraphics[scale=0.3]{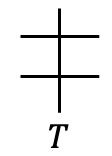}} &=& \sum_j \raisebox{-.5\height}{\includegraphics[scale=0.3]{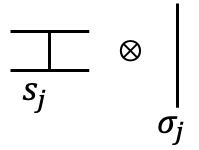}}.
\end{eqnarray}
Note that this is \textit{not} an SVD decomposition. We have simply written $T$ choosing some orthonormal basis in the physical space. At this point we do not know anything about $s_j$ operators. Now taking the trace on the physical legs above gives us $d\mathbb{E}_M$ on the LHS and $d s_0$ on the RHS as all other terms vanish since $\operatorname{Tr}(\sigma_0 ) =d$ and $\operatorname{Tr}(\sigma_{j\neq 0})=0$. So we get, 
\begin{eqnarray}\label{s0EM}
s_0 = \mathbb{E}_M.
\end{eqnarray}
This implies,
\begin{eqnarray}\label{s0N=1}
\operatorname{Tr}(s_0^N) = \operatorname{Tr}(\mathbb{E}_M^N) = 1,
\end{eqnarray}
as per our assumption.

Now $MPO^{(N)\dagger}MPO^{(N)}$ can be expressed in terms of $\sigma_j$ and $s_j$ as 
\begin{eqnarray}\label{ONTsigma}
 \raisebox{-0.4\height}{\includegraphics[scale=0.3]{mpompo.png}}  
&=&\raisebox{-0.4\height}{\includegraphics[scale=0.3]{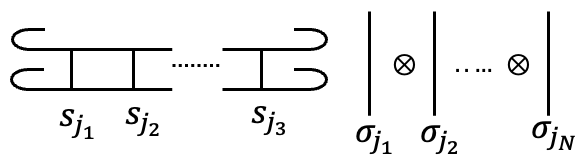}} \nonumber \\
 & = & \sum_J \operatorname{Tr}(s^{(N)}_J) \sigma^{(N)}_J \nonumber\\
&=& \sum_{J = | 00...0 \rangle} I^{\otimes N} +\sum_{J \neq | 00...0 \rangle} \operatorname{Tr}(s^{(N)}_J) \sigma^{(N)}_J,
\end{eqnarray}
where for convenience of representation, we combined indices  $j_1,j_2,\ldots, j_N$ into a vector index $J=|j_1,j_2,\ldots,j_N\rangle $. So $s^{(N)}_J = s_{j_1}s_{j_2}\ldots s_{j_N}$ and $\sigma^{(N)}_J = \sigma_{j_1}\otimes \sigma_{j_2}\otimes \ldots \otimes \sigma_{j_N}$. We used Eq.~\eqref{s0N=1} in the last equality. 
Note  that $\sigma^{(N)}_J $ form an orthonormal basis for operators on $d^{\otimes N}$ space and satisfy $\operatorname{Tr}(\sigma^{(N)}_J \sigma^{(N)\dagger}_K ) = d^N \delta_{J,K}$.
RHS of the above equation shows that, to prove $MPO^{(N)}$ is unitary, we need to prove 
\begin{eqnarray}\label{weneedtoprove}
\operatorname{Tr}(s^{(N)}_J ) = 0, \quad \forall J \neq | 00...0 \rangle.
\end{eqnarray}
We multiply both sides of Eq.~\eqref{ONTsigma}  with their respective Hermitian adjoints, and take the trace on physical indices,
\begin{eqnarray}\label{ETN}
\raisebox{-0.4\height}{\includegraphics[scale=0.3]{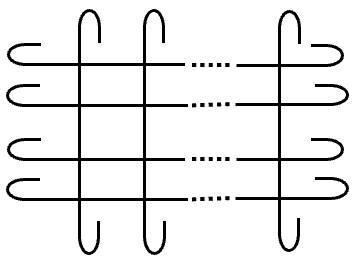}} &=& \sum\limits_{J,K} \operatorname{Tr}(s_J^{(N)})\operatorname{Tr}(s_K^{(N)})^{*} \operatorname{Tr}(\sigma^{(N)}_J \sigma^{(N)\dagger}_K ) \nonumber\\
\Rightarrow \operatorname{Tr}(d^N\mathbb{E}_T^N) &=& d^N \sum_J |\operatorname{Tr}(s_J^{(N)})|^2 \nonumber\\
\Rightarrow \operatorname{Tr}(\mathbb{E}_T^N) &=&\sum_J |\operatorname{Tr}(s_J^{(N)})|^2
, 
\end{eqnarray}
where in the second step we have used the definition of $\mathbb{E}_T$ (Eq. \ref{ETdef}) on the LHS and orthonormality of $\sigma_J^{(N)}$ on the RHS. So now we use the assumption $\operatorname{Tr}(\mathbb{E}_T^N) =1$ in above and get,
\begin{eqnarray}
 1&=&  |\operatorname{Tr}(s_{J=|00\ldots \rangle})^{(N)}|^2 + \sum_{J\neq|00\ldots \rangle} |\operatorname{Tr}(s_{J} ^{(N)})|^2 \nonumber \\ 
\Rightarrow 1 &=& |\operatorname{Tr}(s_0^N)|^2 + \sum_{J\neq|00\ldots \rangle} |\operatorname{Tr}(s_{J} ^{(N)})|^2 \nonumber \\
\Rightarrow 1 &=& 1 +  \sum_{J\neq|00\ldots \rangle} |\operatorname{Tr}(s_{J}^{(N)})|^2 \nonumber \\
\Rightarrow \operatorname{Tr}(s^{(N)}_{J}) &=& 0, \quad \forall J \neq|00\ldots \rangle.\label{trsj0}
\end{eqnarray}
Last equality comes from the fact that RHS is a sum over positive values only. We have proved the key step required (Eq.~\ref{weneedtoprove}) to prove $MPU^{(N)}$-generation. So simply putting Eq.~\eqref{trsj0} in Eq.~\eqref{ONTsigma} we get
\begin{eqnarray}
MPO^{(N)}(M)MPO^{(N)\dagger }(M) &= &  I^{\otimes N}.
\end{eqnarray}
Now we prove $MPO^{(N)\dagger } MPO^{(N)} = I^{\otimes N}$ as well. 
Since $\mathbb{E}_{M^{\dagger}}$ and $\mathbb{E}_{T(M^{\dagger})}$ have the same spectrum as that of $\mathbb{E}_M$ and $\mathbb{E}_{T(M)}$, respectively (they are the same matrices with virtual indices permuted), we also have $\operatorname{Tr}(\mathbb{E}_{M^{\dagger}}^N) =\operatorname{Tr}(\mathbb{E}^N_{T(M^{\dagger})})=1$. So following the same steps as above but for $M^{\dagger}$ we can prove $ MPO^{(N)}(M^{\dagger})MPO^{(N)\dagger }(M^{\dagger}) =  I^{\otimes N}$. Using equality Eq.~\eqref{mdaggermpo}, this gives us $ MPO^{(N)\dagger }(M) MPO^{(N)}(M)=   I^{\otimes N}$. So finally we have proved that $M$ is $MPU^{(N)}$-generating. This completes the proof.
\end{proof}
The $N$-unitarity theorem completely characterizes all possible uniform MPUs. All  uniform MPUs can be seen as a particular solutions to the MPU-generation equation. Conversely, MPU-generation equation gives a convenient and an efficient method (in $O(D^6)$ time) for evaluating if a local tensor generates an MPU of a given size.  
 \par 
Now we turn to locality-preserving aspect of an MPU.
\section{locality-preserving MPUs} \label{sec:revisit LP MPU}
All (uniform) MPUs are solutions to MPU-generation equation (Eq.~\eqref{Nunitarycondition}), but are all of them locality-preserving, as discussed in Ref.~\onlinecite{Burak2018,Cirac2017,Po2016}, as well? We will discuss locality preservation in this section.

 We first define what a locality-preserving unitary operator means for our purpose. 
\begin{mydef}
A unitary operator $U^{(N)}$ acting on a one-dimensional quantum lattice of size $N$ is said to be locality-preserving with range $\nu  \ll N$  if for all sites $x$, $y$ such that   $|x - y| > \nu$, and for all operators $A \in \mathcal{R}_x, B \in \mathcal{R}_y$, we have:
\begin{eqnarray}
    [U^{(N)\dagger} A U^{(N)}  , B] = 0
\end{eqnarray}

where $\mathcal{R}_{x}$ denotes the full operator algebra acting on site $x$.
That is, under unitary evolution by $U^{(N)}$, any local operator remains supported within a finite neighborhood of its original support.
\end{mydef}
We are going to show that under canonical considerations where $M$ generates MPUs on a 1D system for all system sizes, MPUs indeed are locality-preserving. However, without that restriction, there can be solutions to MPU-generator equation that are not locality-preserving (mapping local operators to global operators). 
We will prove this in following two steps: a) $M$ that is an MPU-generator for all system sizes has a transfer matrix $\mathbb{E}_T$ which has exactly one non-zero eigenvalue and it is equal to 1. b) If $\mathbb{E}_T$ has exactly one non-zero eigenvalue  and it is equal to 1, then the corresponding MPU is locality-preserving.

We present a few auxiliary results required for the proof. The first lemma, whose proof can also be found as Lemma 9 in Ref.~\onlinecite{Cuevas2016}, is stated below.
\begin{mylemm}\label{cirac lemm}
Consider two sets of complex numbers,  $\alpha_{a}, a =1,2,\ldots, m;\beta_b, b=1,2,\ldots,n $. If $\forall N \leq max \lbrace m,n \rbrace$ and if the following is true, 
\begin{eqnarray}
\sum_{a=1}^{m} \alpha_a^N= \sum_{b=1}^{n} \beta_b^N,
\end{eqnarray}
then $m=n$ and $\alpha_a=\beta_b$ (up to permutation). 
\end{mylemm}

\begin{mycorol}\label{Cor for LPMPU}
If a local-tensor $M$ generates MPUs for all lengths $N\leq D^4$, then $\mathbb{E}_M$ and $\mathbb{E}_T$ both have exactly one non-zero eigenvalue and it is equal to 1. Moreover, then it also generates MPU for \textit{all} system sizes.
\end{mycorol}
\begin{proof}
Let's say $\mathbb{E}_M$ has eigenvalues $\lambda_1 \geq \lambda_2 \geq \ldots \geq \lambda_{D^2}$. Then, we have
\begin{eqnarray} \label{lambdaN=1}
\sum_j \lambda_j^N = 1, \quad \forall N\leq D^2.
\end{eqnarray}
 Applying Lemma \ref{cirac lemm} on this equation we get $\lambda_1=1, \lambda_2=\lambda_3=\ldots = 0$. So $\mathbb{E}_M$ has exactly one non-zero eigenvalue and it is equal to 1. The same can be applied to $\mathbb{E}_T$ and we get that $\mathbb{E}_T$ has exactly one non-zero eigenvalue and it is equal to 1. This also means that $\operatorname{Tr}(\mathbb{E}^N_M)=\operatorname{Tr}(\mathbb{E}^N_T)=1, \forall N\geq 1$, which implies that it generates MPUs for all system sizes. This completes the proof. 
\end{proof}
We also record the following standard identity for completeness, as it will be used later.
\begin{mylemm}\label{unitary lemm}
If $A$ and $B$ are two unitary operators of dimension $d$ and if $\operatorname{Tr}(AB A^\dagger B^\dagger )= d$, then $A$ and $B$ commute. 
\end{mylemm}
\begin{proof}
First note that $\operatorname{Tr}\left([A,B][A,B]^\dagger \right) = 0 \Rightarrow [A,B] =0$.
Now,
\begin{eqnarray}
    \operatorname{Tr}\left([A,B][A,B]^\dagger \right)  &=& \operatorname{Tr}( ABB^\dagger A^\dagger) + \operatorname{Tr}( B AA^\dagger B^\dagger) - \operatorname{Tr}( A BA^\dagger B^\dagger) - \operatorname{Tr}(BAB^\dagger A^\dagger ) \nonumber \\
    &=& \operatorname{Tr}(I) + \operatorname{Tr}(I) -\operatorname{Tr}\left(  A BA^\dagger B^\dagger\right) -\operatorname{Tr}\left( A BA^\dagger B^\dagger)^\dagger\right)   \nonumber \\
    &=& d + d -d -d =0.
\end{eqnarray}

\end{proof}
Now we present the main result for locality-preserving MPUs.
\begin{mythm}\label{MPUO are LP}
A local tensor that generates MPU on all system sizes is also locality-preserving.
\end{mythm}

\begin{proof}
We have already proved that if a tensor generates MPU on all system sizes then $\mathbb{E}_T$ has exactly one non-zero eigenvalue and it is equal to one. The eigenvalues 0 may have a Jordan block associated to them. So there must exist a number $v < D^4$, where $D$ is the dimension of the virtual legs, such that  $\mathbb{E}_T^r, \, \forall r\ge \nu$, has no Jordan blocks, and hence is a rank 1 matrix. Any rank 1 matrix can be written as an outer product of two vectors, so
\begin{eqnarray}\label{eq: Evm=RL}
\mathbb{E}_T^r &=& | R \rangle \langle L|, \quad \forall r \ge \nu\nonumber\\
\textrm{with } \langle L| R \rangle &=& 1
\end{eqnarray}
or pictorially,
\begin{eqnarray}\label{eq:pic Evm=RL}
\raisebox{-0.55\height}{\includegraphics[scale=0.22]{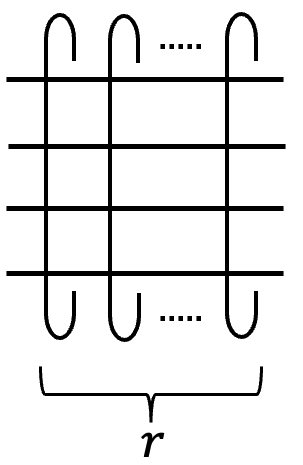}}   &=& d^{r}\raisebox{-0.4\height}{\includegraphics[scale=0.2]{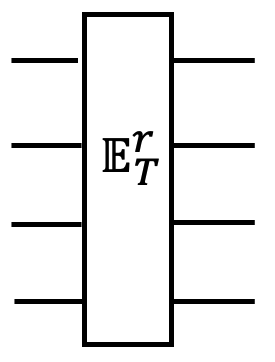}}  =d^{r}\raisebox{-0.4\height}{\includegraphics[scale=0.2]{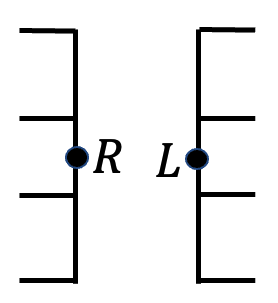}}.
\end{eqnarray}

 Now we will show how this decomposition of $\mathbb{E}_T$ directly implies locality preservation.

  We will work with Pauli matrices as a basis for local operator algebra, though any unitary basis can be used. Let $\sigma_{x;n} $  denote the $n$th Pauli matrix on site $x$, with $n \in [0,d^2-1]$, and identity on all other sites. To prove locality preservation, we want to prove that
  \begin{eqnarray}
      [  U^{(N)\dagger} \sigma_{x;n} U^{(N)}  , \sigma_{y;m}] = 0,
  \end{eqnarray}
  for all $x,y$ with $|x-y|>\nu$ and for all basis $\sigma_{x;n},\sigma_{y;m}$.
 So, using $A=U^{(N)\dagger} \sigma_{x;n} U^{(N)}$ and $B = \sigma_{x;m}$, both unitary operators of dimension $d^N$,  in Lemma~\ref{unitary lemm}, we need to prove that
  \begin{eqnarray}
      \operatorname{Tr}\left( U^{(N)\dagger} \sigma_{x;n} U^{(N)}\sigma_{y;m} U^{(N)\dagger} \sigma_{x;n}^\dagger U^{(N)}\sigma_{y;m}^\dagger \right) = d^N.
  \end{eqnarray}
  \begin{eqnarray}
      \raisebox{-0.6\height}{\includegraphics[scale=0.25]{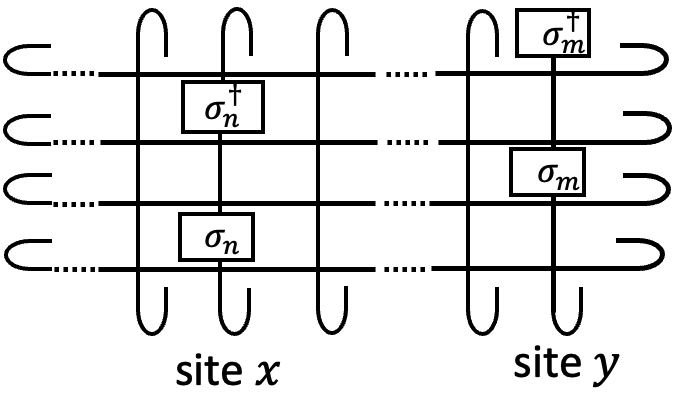}} &=& d^{N-2}
 \raisebox{-0.6\height}{\includegraphics[scale=0.25]{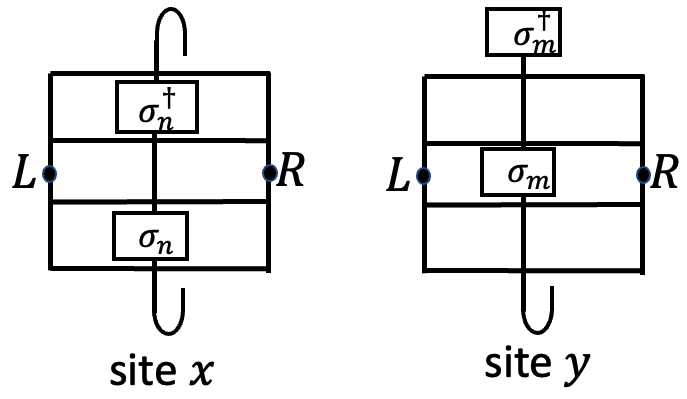}}
  \end{eqnarray}
 LHS can be represented with tensor network as follows:
\begin{eqnarray}
       \operatorname{Tr}\left( U^{(N)\dagger} \sigma_{x;n} U^{(N)}\sigma_{y;m} U^{(N)\dagger} \sigma_{x;n}^\dagger U^{(N)}\sigma_{y;m}^\dagger \right)   &=&\raisebox{-0.6\height}{\includegraphics[scale=0.25]{support-new.png}} \nonumber \\ 
  &=& d^{N-2}
 \raisebox{-0.6\height}{\includegraphics[scale=0.25]{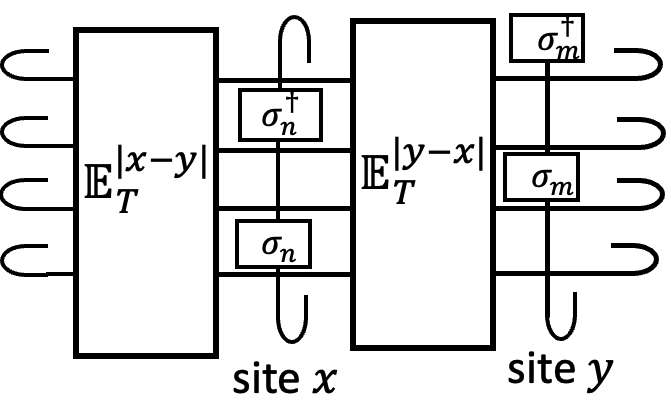}}\nonumber \\
   &=&d^{N-2}
 \raisebox{-0.6\height}{\includegraphics[scale=0.25]{support-new-3.png}}, \label{eq:RHS}
\end{eqnarray}
where we used Eq.~\eqref{eq:pic Evm=RL} in the second and third equality. A small subtlety should be noted here. In Eq.~\eqref{eq:pic Evm=RL}, the contraction order is $\operatorname{Tr}(MM^\dagger MM\dagger)$, whereas in the expression above it is $\operatorname{Tr}(M^\dagger MM\dagger M)$. Thus, while $\mathbb{E}_T$, $|R\rangle$ and $\langle L |$ refer to the same transfer matrix and fixed points in both cases, their virtual indices are effectively shifted by one position.

The two tensor contractions on the RHS of the last equality are now independent of each other and can be easily calculated. Again, using tensor network representation, followed by $\mathbb{E}_T$ decomposition as above, we can easily show that
\begin{eqnarray}
     \operatorname{Tr}\left( U^{(N)\dagger} \sigma_{x;n} U^{(N)} U^{(N)\dagger} \sigma_{x;n}^\dagger U^{(N)}    \right)   &=& \raisebox{-0.5\height}{\includegraphics[scale=0.2]{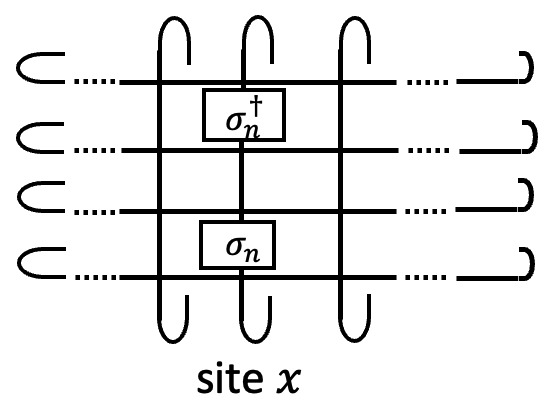}} \nonumber \\
    \Rightarrow d^N &=&d^{N-1} \raisebox{-0.5\height}{\includegraphics[scale=0.2]{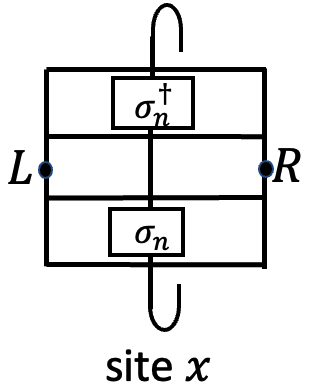}}.
\end{eqnarray}
This shows that the value of the first contraction on the RHS of eq~\eqref{eq:RHS} is simply $d$. Similarly, the value of the second contraction on the RHS of eq~\eqref{eq:RHS} can be shown to be $d$ as well. Putting it all together, we get 
\begin{eqnarray}
         \operatorname{Tr}\left( U^{(N)\dagger} \sigma_{x;n} U^{(N)}\sigma_{y;m} U^{(N)\dagger} \sigma_{x;n}^\dagger U^{(N)}\sigma_{y;m}^\dagger \right) = d^{N-1}.d.d = d^N.
\end{eqnarray}
 This completes the proof.
\end{proof}
\section{Examples}\label{sec:examples}
In this section we discuss some specific examples and how the results from this work apply to them. First we take the canonical right translation tensor,
\begin{equation}
    M_r= \raisebox{-.4\height}{\includegraphics[height=1.2cm]{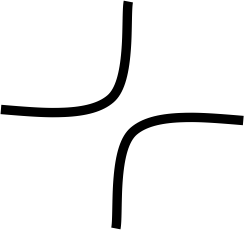}}.
\end{equation}
Its transfer-matrices can be calculated quite simply to be,
\begin{eqnarray}
\mathbb{E}_M &=& \frac{1}{d}\raisebox{-.4\height}{\includegraphics[height=1.2cm]{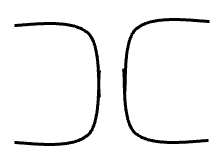}},
\end{eqnarray}
and
\begin{eqnarray}
\mathbb{E}_T &=& \frac{1}{d}\raisebox{-.4\height}{\includegraphics[scale=0.2]{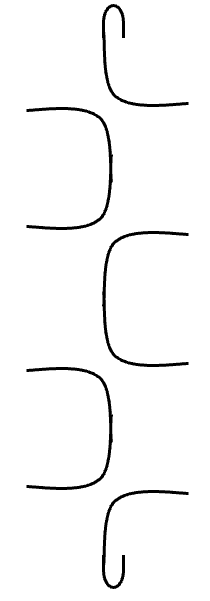}}.
\end{eqnarray}
It can be immediately seen that both are rank 1 matrices already written in $|R\rangle \langle L|$ form with $ \langle L|R \rangle=1$. And hence they satisfy condition~\eqref{Nunitarycondition} for all $N$ and hence they generate locality-preserving MPUs. A same result can be shown for left translation tensor. \\
Now we take a non-trivial example presented originally in Ref.~\onlinecite{Burak2018},
\begin{eqnarray}\label{ex1011}
M = \sum_{a,b=0}^2 \raisebox{-.4\height}{\includegraphics[scale=0.3]{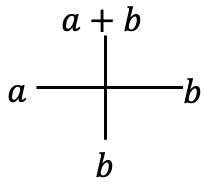}},
\end{eqnarray}
where $(a+b) \equiv  (a+b) \text{ mod } 3$. It was shown to generate MPU for only odd system sizes and was a non-locality-preserving MPU. 
 Its transfer matrix can be calculated  to be 
\begin{eqnarray}
\mathbb{E}_M = \frac{1}{3}\raisebox{-.4\height}{\includegraphics[scale=0.3]{EMr.png}},
\end{eqnarray}
which has only one non-zero eigenvalue which is equal to 1. So we see that transfer matrix for this $M$ is exactly similar to those of locality-preserving translation MPUs. But what makes it different is the matrix $\mathbb{E}_T$: 
\begin{eqnarray}
\mathbb{E}_T = \frac{1}{3}\sum_{a,b=0}^2  \delta_{a_1+b_1,a_2+b_2}\delta_{a3+b2,a4+b1}\raisebox{-.4\height}{\includegraphics[scale=0.3]{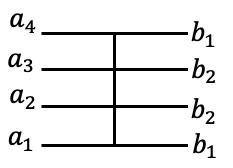}}.
\end{eqnarray}
By calculating its eigenvalues, we find that the only non-zero eigen values are $\lbrace 1,1,-1 \rbrace $. So if it is to satisfy the condition of $MPU$-generation theorem~\ref{main theorem}, we should have
\begin{eqnarray}
\operatorname{Tr}(\mathbb{E}_T^N)=1^N+1^N+(-1)^N=1,
\end{eqnarray}
which is true if and only if $N$ is odd. This explains why this MPU is unitary  for only odd system sizes. And since it does not generate unitary for all system sizes we do not expect it to be locality-preserving and indeed it is not. This example also shows that the causal cone of MPU's operation is determined by the structure of $\mathbb{E}_T$ and not $\mathbb{E}_M$.  This goes to show that there are many non-trivial solutions to $MPU$-generation theorem theorem beyond locality-preserving MPUs. 
\section{Conclusion}\label{sec:conclusion}
In this work, we have established a simple and efficient method for determining whether a given local tensor generates an MPU of size $N$. The central result, expressed through the $N$-unitarity theorem, shows that the conditions \[
\operatorname{Tr}(\mathbb{E}_M^N) = \operatorname{Tr}(\mathbb{E}_T^N) = 1\] are both necessary and sufficient for a tensor $M$ to generate an MPU of length $N$. This provides a unified and computationally efficient approach for identifying uniform MPUs, bypassing the cumbersome fixed-point conditions employed in previous works.

Furthermore, we demonstrated that when a tensor generates MPUs for all system sizes, it necessarily ensures locality preservation. This connection highlights how the same transfer matrix framework can capture both unitary and locality-preserving properties, offering a deeper understanding of the interplay between dynamical processes and entanglement structure in 1D quantum systems.

Our method not only simplifies the universal understanding of uniform MPUs but also provides insights into the conditions under which non-locality-preserving MPUs emerge. Through numerical examples, we illustrated that while locality-preserving MPUs satisfy the MPU-generation conditions for all sizes, certain non-trivial solutions can arise when the conditions are only partially satisfied (e.g., for specific sizes). These findings emphasize the rich and diverse structure of MPUs beyond the locality-preserving subclass.

It should be noted that there is nothing in the main theorem~\ref{main theorem} of this work that needs the assumption of MPO being in 1D. This proof and results are easily generalizable to higher dimensions. However, 2D traces of transfer matrices are not efficiently calculated and, in general,  will scale with system size. However it remains an open question whether a sufficient condition can still be derived purely based on the structure of the transfer matrix. 

Different classes of MPUs are nothing but different solutions to the equations ~\ref{main theorem}. However we did not explore in this work how those classes, with different GNVW indices, are differentiated. For example, we saw that two different classes of MPUs can still have the same $\mathbb{E}_M$. But their $\mathbb{E}_T$ were different from each other. This should be explored further if it is possible to classify MPUs only based in these transfer matrices. 

In summary, this work contributes a new perspective on the efficient characterization of MPUs, offering practical tools for further exploration of their role in simulating and understanding complex quantum phenomena, including the classification of topological phases and quantum cellular automata. Future studies could extend these results to higher-dimensional systems and investigate connections to non-uniform and symmetry-enriched MPUs.
\section{ACKNOWLEDGMENT}
The author thanks M. Burak {\c S}ahino{\u g}lu and Xie Chen for helpful discussions. This work was supported by the Walter Burke Institute for Theoretical Physics and the Institute for Quantum Information and Matter, an NSF Physics Frontiers Center (NSF Grant No. PHY-1125565), with support from the Gordon and Betty Moore Foundation (GBMF-2644).
 \bibliography{TNORef.bib}
 \end{document}